\documentclass[10pt,journal]{article}

\usepackage{amsthm}
\usepackage{wrapfig}
\usepackage{mathptmx}
\usepackage{multirow}
\usepackage{url}
\usepackage{epstopdf,amssymb} 
\usepackage{lipsum}
\usepackage{graphicx}
\usepackage{amsmath,bm}
\usepackage{algorithm}
\usepackage{algpseudocode}

\usepackage{lineno}

\newcommand{\captionsize}{\footnotesize}

\newtheorem{thm}{Theorem}

\newtheorem{proposition}{Proposition}
\newtheorem{lem}{Lemma}

\begin{document}
\title{Mending Missing Information in Big-Data}
%
%
%

\author{Hadassa~Daltrophe,
        Shlomi~Dolev
        and~Zvi~Lotker}

\maketitle



\begin{abstract}
Consider a high-dimensional data set, in which for every data-point there is incomplete information. Each object in the data set represents a real entity, which is described by a point in high-dimensional space. We model the lack of information for a given object as an affine subspace in $\mathbb{R}^d$ whose dimension $k$ is the number of missing features.

Our goal in this study is to find clusters of objects where the main problem is to cope with partial information and high dimension. Assuming the data set is separable, namely, its emergence from clusters that can be modeled as a set of disjoint ball in $\mathbb{R}^d$, we suggest a simple data clustering algorithm. Our suggested algorithm use the affine subspaces minimum distance and 
calculates pair-wise projection of the data achieving poly-logarithmic time complexity. 

We use probabilistic considerations to prove the algorithm's correctness. These probabilistic results are of independent interest, and can serve to better understand the geometry of high dimensional objects.
\end{abstract}


\section{Introduction}
One of the main challenges that arise while handling Big-Data is not only the large volume, but also the high-dimensions of the data. Moreover, part of the information at the different dimensions may be missing. Assuming that the true (unknown) data is $d$-dimensional points, we suggest representing the given data point (which may lack information at different dimensions) as a $k$-affine space embedded in the Euclidean $d$ dimensional space $\mathbb{R}^d$. Denote the affine-Grassmannian set of all $k$-affine spaces, embedded in the Euclidean $d$ dimensional space, as $A(d,k)$. This means that a point in our data set is a point in the affine-Grassmannian $A(d,k)$.    

A data object that is incomplete in one or more features corresponds to an affine subspace (called flat, for short) in $\mathbb{R}^d$, whose dimension is the number of missing features. This representation yields algebraic objects, which help us to better understand the data, as well as study its properties. A central property of the data is clustering. Clustering refers to the process of partitioning a set of objects into subsets, consisting of similar objects. Finding a good clustering is a challenging problem. Due to its wide range of applications, the clustering problem has been investigated for decades, and continues to be actively studied not only in theoretical computer science, but in other disciplines, such as statistics, data mining and machine learning. A motivation for cluster analysis of high-dimensional data, as well as an overview on some applications where high-dimensional data occurs, is given in~\cite{Kriegel2009}. 

Our underlying assumption is that the original data-points, the real entities, can be divided into different groups according to their distance in the $\mathbb{R}^d$. We assume that every group of points lie in the same $d$ dimensional ball $\mathbb{B}^d$ (a.k.a. a solid sphere), since the distance between a flat and a point (the center of the ball) is well-defined. The classic clustering problems, such as $k$-means or $k$-centers (see~\cite{Hopcroft2014} Chapter $8$), can be defined on a set of flats. The clustering problem when the data is $k$-flats, is to find the centers of the balls that minimizes the sum of the distance between the $k$-flats and the center of their groups, which is the nearest center among all centers.

However, Lee $\&$ Schulman~\cite{Lee2013} argues that the running time of an approximation algorithm, with any approximation ratio, cannot be polynomial in even one of $m$ (the number of clusters) and $k$ (the dimension of the flats), unless $P = NP$. We overcome this obstacle by approaching the problem differently. Using a probabilistic assumption based on the distribution of the data, we achieve a polynomial algorithm, which we use to identify the flats' groups. Moreover, the presented probability arguments can help us in better understanding the geometric distribution of high dimensional data objects, which is of major interest and importance in the scope of Big Data research.

\subsection*{Our contributions.}
We face the challenge of mending the missing information at different dimensions by representing the objects as affine subspaces. In particular, we work within the framework of flat in $\mathbb{R}^d$, where the missing features correspond to the (intrinsic) dimension of the flat. This representation is accurate and flexible, in the sense that it saves all the features of the origin data; it also allows for algebraic calculation over the objects. In this chapter, we study the pairwise distance between the flats, and based on our probabilistic and geometrical results, we developed a polylogarithmic algorithm that achieves clustering of the flats with high probability. 

The main result of the study is summarized in the following theorem, while the precise definition and the detailed proof are presented in the sequel.
\begin{thm} 
\label{thm:main}
Given the separable data set $\mathbf{P}$ of $n$ affine subspaces in $\mathbb{R}^d$, for any $\epsilon > 0$ and for sufficiently large $d$ (depending on $\epsilon$), with probability $1-\epsilon$, we can cluster $\mathbf{P}$ according to $\mathbb{B}^d$, using their pair-wise distance projection in $poly(n,k,d)$ time.
\end{thm}
\textit{Remarks:}
\begin{itemize}
	\item In addition to proving good performance for high dimensions as required in the scope of big-data, we also show that the algorithm works well for low dimensions.
	\item Using sampling, we achieve a \textit{poly-logarithmic} running time. 
	\item We show we can relax the model assumption about the identical size of clusters to any different sizes.
\end{itemize}
 
To enhance the readability of our text, Section~\ref{sec:pre} contains the basic notions, from convex and stochastic geometry, which are needed in the following. In particular, we recall the notion of flats and provide the model assumptions. We prove our main result in Section~\ref{sec:flatClustering}, and summarized the suggested Algorithm~\ref{algo:1} in Section~\ref{sec:algo}. We supplement our theoretical results with experimental data in Section~\ref{sec:exp}, and generalize our results to clusters with different size in Sections~\ref{sec:diffGroupSize}. In Section~\ref{sec:subAndDistribute} we illustrate how one can change our algorithm to work in sublinear time and to implement it in distributed fashion. Finally, in Section~\ref{sec:diss}, we discuss the geometric and algebraic representation, comparing our approach against others' proposals. 

\section{Preliminaries}
\label{sec:pre}
\subsection*{General notation.} Throughout the following, we work in $d$-dimensional Euclidean space $\mathbb{R}^d$,
$d \geq 2$, with scalar product $\left\langle \cdot,\cdot\right\rangle$ and norm $\left\|\cdot\right\|$. Hence, $\left\|x-y\right\|$ is the Euclidean distance of two points $x,y\in\mathbb{R}^d$, and $dist(X,Y):=\inf\left\{\left\|x-y\right\|:x\in X, y\in Y\right\}$ is the distance of two sets $X, Y \in\mathbb{R}^d$. We refer to any set $S\subseteq X$, which is closest to $Y$, i.e., satisfies $\left\|Y-S\right\| = dist(X,Y)$, as a \textit{projection} of $Y$ on $X$. In general, there can be more than one projection of $Y$ on $X$, i.e., several subsets in $Y$ closest to $X$. 

\subsection*{Grassmannians.} 
For $k\in\left\{1,...,d\right\}$, we denote by $G(d, k)$ and $A(d, k)$ the spaces of $k$-dimensional linear and affine subspaces of $\mathbb{R}^d$, respectively, both supplied with their natural topologies (see e.g.,~\cite{Schneider2008}). The elements of $A(d, k)$ are also called $k$-flats (for $k = 0$, points; for $k = 1$, lines; for $k = 2$, planes; and for $k = d-1$, hyperplanes). Recall that two subspaces $L\in G(d, k_1)$ and $M\in G(d, k_2)$ are said to be in \textit{general position} if the span of $L\cup M$ has dimension $k_1 + k_2$ if $k_1 + k_2 < d$ or if $L \cap M$ has dimension $k_1 + k_2 - d$ if $k_1 + k_2 \geq n$. We also say that two flats $E \in A(d, k_1)$ and $F \in A(d, k_2)$ are in general position, if this is the case for $L(E)$ and $L(F)$, where $L(E)$ is the linear subspace parallel to $E$.

\subsection*{Geometric and Probabilistic definitions.}

Let $\mathbf{P}=\left\{P_1,P_2,...,P_n\right\}$ be the set of $n$ random flats that we want to cluster. For the sake of simplicity, we consider the situation where all of them are of dimension $k$, where $k$ is taken to be the greatest dimension of any flat in $P$. Hence, every flat $P$ is represented by a set of $d-k$ linear equations, each with $d$ variables. Alternatively, we can represent any $k$-flat using a parametric notation, such that $P$ is given by a set of $d$ linear equations, each with $d-k$ variables.

When there is no flat with a fixed $i$th coordinate, we will call the $i$th coordinate \textit{trivial}. We can assume that \textit{no} coordinate is trivial, since otherwise, simply removing this coordinate from all flats will decrease $k$ and $d$ by $1$, while not affecting the clustering cost.

For $c\in \mathbb{R}^d$, let $\mathbb{B}^d_{c}$ be the unit ball of dimension $d$, centered at $c$ and $\mathbb{B}^d_{c_0}$ denote the unit ball centered at the origin. Two balls, $\mathbb{B}_{c_i}$ and $\mathbb{B}_{c_j}$, are $\Delta$-\textit{distinct} if $dist(c_i,c_j)\geq\Delta$. The ball $\mathbb{B}^d_{c}$ intersects the subset of flats $P=\left\{P_1,...,P_j\right\}$ if it intersects each flat in $P$. We will denote by $P^{c}_i\in\mathbf{P}$ a $k$-flat intersecting the unit ball $\mathbb{B}^d_{c}$ and by $P_i(r)\in\mathbf{P}$ a $k$-flat in $\mathbb{R}^d$ passing through the point $\left(r,0,...,0\right)$. 

Let $*$ be an equivalence relation such that for a point $u\in \mathbb{B}^d_c$, $u^*$ is the antipodal point of $u$ (i.e., $u$ and $u^*$ are opposite through the center $c$). For a $k$-flat $P^{c}$ intersecting the unit ball $\mathbb{B}^d_c$ in one point only (i.e., tangent to the balls surface), $P^{c*}$ denote its antipodal $k$-flat.

If $E$ and $F$ are in general position, there are unique points $x_E\in E$ and $x_F\in F$, so that $dist(E, F) = \left\|x_E-x_F\right\|$. We call the point $p=midpoint\left(E,F\right):=\left(x_E+x_F\right)/2$ the \textit{midpoint} of $E$ and $F$.

The probability, expectation and variance; will be denoted by the common notations $\operatorname{Pr}(\cdot),\operatorname{E}(\cdot)$ and $\operatorname{V}(\cdot)$ respectively. For a random variable $A$ dependent on $d$, we denote by $A\to_p c$ the ``converges in probability'', namely, $\forall\epsilon, \lim\limits_{d\to\infty}\operatorname{Pr}(\left\|A-c\right\|\leq \epsilon)=1$.

\subsection*{Model assumptions.}
Throughout the chapter we assume that the data is \textit{separable}, namely, satisfing the following assumptions:
\begin{itemize}
\item Two independent random flats $E,F\in A(d,k)$, with distribution $\mathbb{Q}$, are in general position with probability one.
\item $1\leq k\leq \left\lfloor d/2\right\rfloor$ which ensures that the flats do not intersect each other with probability one.
\item The (unknown) balls $\mathbb{B}^d_{c_1},...,\mathbb{B}^d_{c_m}$ are $\Delta$-\textit{distinct} with probability one.
\item The given flats set $\mathbf{P}$ is a superset of $m$ groups $\mathbf{P}=\left\{P_1,...,P_m\right\}$, such that every group $P_i\in\mathbf{P}$ contains $n/m$ flats that intersect the ball $\mathbb{B}^d_{c_i}$. Moreover, each flat $P\in P_i$ has a \textit{normally distributed} location and direction at the ball $\mathbb{B}^d_{c_i}$. We model this assumption by normally distributed coefficients. The parametric representation of a $k$-flat $P$ is: 
$$P=At+a= \left( \begin{array}{ccc}
\alpha_{0,1}+\alpha_{1,1}t_1+\alpha_{2,1}t_2+...+\alpha_{k,1}t_k \\
... \\
\alpha_{0,d}+\alpha_{1,d}t_1+\alpha_{2,d}t_2+...+\alpha_{k,d}t_k\\
\end{array} \right)
$$
where $\alpha_{ij}\sim N(\mu,\sigma)$ and $t$ is the $k$-dimensional vector. 
\end{itemize}
\section{$k$-flats Clustering}
\label{sec:flatClustering}

Given the set $\mathbf{P}$ of $n$ $k$-flats in $\mathbb{R}^d$, our goal is to cluster the flats according to the unknown set of balls, namely, to separate $\mathbf{P}$ into $m$ groups such that every group $P_i\in\mathbf{P}$ contains $n/m$ flats that intersect the same unit ball $\mathbb{B}^d_{c_i}$. We suggest the following procedure (summarized below in Algorithm~\ref{algo:1}) for the clustering process. The first step is to find the distance and the midpoint between every pair of flats in $\mathbf{P}$. Next, we filter the irrelevant midpoints using their corresponding distances such that midpoints with a distance greater than two are dropped and those with a distance $\leq 2$ are grouped together. In the final step we check which group contains $O(n/m)$ flats and output those groups. We argue that these simple steps provide the expected clustering procedure with high probability. In this section, we claim its correctedness using geometric and probabilistic arguments which appear in the following Propositions and Lemmas.

As mentioned above, we start our procedure by calculating the pair-wise projection of $\mathbf{P}$, namely, finding the distance and the midpoint between every pair in $\mathbf{P}$. Let $P_i=\left\{x\in\mathbb{R}^d:Ex=e\right\}$ and $P_j=\left\{y\in\mathbb{R}^d:Fy=f\right\}$ be a pair of $k$-flats in $\mathbf{P}$. Note that the matrices dimensions $Dim(E)=Dim(F)=\left(d-k\right)\times d$ since each flat $P\in \mathbf{P}$ is represented by $d-k$ equations with $d$ variables. The suggested algorithm calculates the minimum distance points (i.e., midpoint) between the pair using Euclidean norm minimization:
\begin{eqnarray}
minimize\:\left\|Ax-b\right\|
\end{eqnarray}
where
$
 A=\begin{pmatrix}
  E \\
  F \\
 \end{pmatrix},x=\left(x_1,...,x_d\right),
 b=\begin{pmatrix}
 e\\
f \\
\end{pmatrix}
$

Since the norm is always nonnegative, we can just as well solve the least squares problem 
\begin{eqnarray}
\label{eq:LS}
minimize \left\|Ax-b\right\|^2
\end{eqnarray}
The problems are clearly equivalent, while the objective in the first one is not differentiable at any $x$ with $Ax - b = 0$, whereas the objective in the second is differentiable for all $x$. 

\begin{proposition}
\label{pro:midpoint}
The least squares minimization (Eq.~\ref{eq:LS}) gives unique solution $p$ such that $p=midpoint(P_i,P_j)$. 
\end{proposition}

\begin{proof}
Using the equation $$minimize\:\left\|Ax-b\right\|^2=\left(Ax-b\right)^T\left(Ax-b\right)$$ this problem is simple enough to have a well known analytical solution - a point $p$ minimizes the function $f=x^TA^TAx-2x^TA^Tb+b^Tb$ if and only if
$$
\nabla f=2A^TAp-2A^Tb=0
$$
i.e., if and only if $p$ satisfies normal equations
$$
A^TAp=A^Tb
$$
which always will have a solution (note that the system is square or over-determined since $2(d-k)\geq d$ for $1\leq k\leq d/2$). The columns in $A$ are the different coordinates of the two flats, hence they are independent and have a unique solution:
$p=(A^TA)^{-1}A^Tb.$\qed\end{proof}

\begin{proposition}
\label{pro:dist}
Using the midpoint $p=midpoint(P_i,P_j)$ one can find the distance between the two flats $dist\left(P_i,P_j\right)$. 
\end{proposition}

\begin{proof}
Theorem $1$ in~\cite{Gross1996} calculates the Euclidean distance between the two affine subspaces using the matrices range and null space. Alternatively, since we already have the midpoint $p$ between the flats we can find the distance between them by projecting $p$ onto the flats and then calculating the distance between the projected points. This projection can be made by a least squares method with constraints, more precisely, to solve the following two optimization problems: $min\left\{\left\|p-x\right\|^2:Ex=e\right\}$ and $min\left\{\left\|p-x\right\|^2:Fx=f\right\}$ or any other efficient orthogonal projection method (e.g.~\cite{Plesnik2007}).
\qed\end{proof}

Having the midpoint and the distance between all the pairs, we filter the irrelevant midpoints using their corresponding distances as shown in the following Lemmas.
First we argue that the flats' pairwise projection helps to define the origin balls, namely, the midpoints that arise from the same ball are centered around that ball:
\begin{lem}
\label{lem:same}
Let $P=\left\{P^c_1,P^c_2,...,P^c_j\right\}\subseteq \mathbf{P}$ be a set of $k$-flats in $\mathbb{R}^d$ intersecting the ball $\mathbb{B}^d_c$. Let $\mathbf{p}=\left\{p_{12},p_{13},...p_{1j},...,p_{\left(j-1\right)j}\right\}$ be the set of the midpoints of all $\binom{j}{2}$ pairs of $P$. The mean of this set $\operatorname{E}[\,\mathbf{p}]$ equals to $c$ (the center of $\mathbb{B}^d_c$), and the variance $\operatorname{V}[\,\mathbf{p}]$ is bounded.
\end{lem}

\begin{proof}
Let $P^c_i,P^c_j\in P$ be two flats intersecting the ball $\mathbb{B}^d_c$ where their distance midpoint is $p_{ij}$. Denote by $p^*_{ij}$ the \textit{antipodal} point of $p_{ij}$. Since the directions and the location of flats at $P$ are normally distributed around $c$ (see the model assumptions at Section~\ref{sec:pre}), we get the probability that $p_{ij}\in \mathbf{p}$ equals to the probability that $p^*_{ij}\in\mathbf{p}$, which implies that their expected value is $\operatorname{E}[\,\left\{p_{ij},p^*_{ij}\right\}]=c$. This geometric-probabilistic consideration holds to the whole set $\mathbf{p}$, hence, we get that $\operatorname{E}[\,\mathbf{p}]=c$.

For proving that the variance is bounded we argue in Proposition~\ref{pro:2Dintersection} and~\ref{pro:Dintersection} (appear at the end of this section) that for all $i,j$, the distance $r_{ij}$ between $p_{ij}$ and the center of the ball $c$ is bounded, which implies that $\operatorname{V}[\,\mathbf{p}]$ is bounded around $c$.    
\qed\end{proof}


At this point, for every pair of flats $(P_i,P_j)$ we have the corresponding midpoint and the distance $(p_{ij},d_{ij})$. We would like to show that if we eliminate all the midpoints $p_{ij}$ so that their distance $d_{ij}$ is greater than $2$, we are left with those that arise from the same cluster. The following Lemma argues that this is the case when $d$ is big enough:

\begin{lem}
\label{lem:diff}
Let $P_i,P_j\in \mathbf{P}$ be a pair of $k$-flats in $\mathbb{R}^d$. 
\begin{enumerate}
	\item If $P_i$ and $P_j$ intersecting the same ball $\mathbb{B}^d_c$ then the probability that the distance between them is less then $2$ is $P\left(dist(P_i,P_j)\leq 2\right)=1$.
	\item Otherwise, for any $\epsilon>0$, $\lim\limits_{d\to\infty}\operatorname{Pr}\left(dist(P_i,P_j)\geq 2\left(\Delta-\epsilon\right) \right)=1$.
\end{enumerate} 
\end{lem}

\begin{proof}
When both flats are intersecting the same unit ball, the minimum distance between them is $\leq 2*radius\left(\mathbb{B}^d_c\right)=2$ which implies the first part of the lemma.
Applying Proposition~\ref{pro:diff} with $dist(P_i,Q_i)\leq 2$ (by the first part of the Lemma), we get that for any $\epsilon$ the distance between the two flats approach $2\left(\Delta-\epsilon\right).$
\qed\end{proof}

\begin{proposition}
\label{pro:diff}
Let $P_{i},Q_{i}$ and $R_j$ be flats intersecting the $\Delta-distinct$ balls $\mathbb{B}_{c_i}$ and $\mathbb{B}_{c_j}$ (respectively).
Then, for any $\epsilon>0$, $\lim\limits_{d\to\infty}\operatorname{Pr}\left(dist\left(R_j,Q_j\right)\geq\left(\Delta-\epsilon\right)dist\left(P_i,Q_i\right)\right)=1$.
\end{proposition}
Note: This proposition appears in~\cite{Bennett1999} for random points. Here we reproduce a proof for the distance between the flats.

\begin{proof} 
Let $\mu=\operatorname{E}\left(dist\left(P_i,Q_i\right)\right)$, $V=\frac{dist\left(P_i,Q_i\right)}{\mu}$ and $W=\frac{dist\left(R_j,Q_i\right)}{\mu}$. 
Using Lemma~\ref{lem:same} and the weak law of large numbers we get that $V\to_p 1$. Proposition~\ref{pro:diff_mean} implies that $W\to_p \Delta$. Thus, $\frac{dist(R_j,Q_i)}{dist(P_i,Q_i)}=\frac{\mu dist(R_j,Q_i)}{\mu dist(P_i,Q_i)}=\frac{W}{V}\to_p\Delta$ (see Corollary 1 at~\cite{Beyer1999}). By definition of convergence in probability for any $\epsilon>0$, $\lim\limits_{d\to\infty}\operatorname{Pr}\left(\left|\frac{dist(R_j,Q_i)}{dist(P_i,Q_i)}-\Delta\right|\leq \epsilon\right)=1$. So $\lim\limits_{d\to\infty}\operatorname{Pr}\left(\Delta-\epsilon\leq \frac{dist(R_j,Q_i)}{dist(P_i,Q_i)}\leq\Delta +\epsilon\right)=1$ which implies $\lim\limits_{d\to\infty}\operatorname{Pr}\left(dist\left(R_j,Q_j\right)\geq\left(\Delta-\epsilon\right)dist\left(P_i,Q_i\right)\right)=1.$
\qed\end{proof} 

Lemma~\ref{lem:diff} implies the correctness of our algorithms when $d\to\infty$. The following Propositions argue that for \textit{any} dimension $d$, when we dropped the midpoints with corresponding distances $\leq 2$ we eliminate at least a linear fraction $\lambda$ of the whole set. Proposition~\ref{pro:diff_mean} show that mean distance of flats is linear with respect to the distance. Next, we use this result to prove that we drop enough flats as presented in Proposition~\ref{pro:prob_drop}.

\begin{proposition}
\label{pro:diff_mean}
Let $P_{i}$ and $P_j$ be flats intersecting the $\Delta-distinct$ balls $\mathbb{B}_{c_i}$ and $\mathbb{B}_{c_j}$ (respectively), then $\operatorname{E}\left(dist\left(P_i,P_j\right)\right)$ is linear function of $\Delta$.
\end{proposition}

\begin{proof}
Denote the mean distance integral between two $k$-flats in $\mathbb{R}^d$ by $S=\operatorname{E}\left(dist(P_i,P_j)\right)$. Given that the probability density function of the flats is $\rho$, the expected value of the distance function, is given by the inner product of the functions $dist$ and $\rho$. E.g., for the $d$ dimensional lines $P(1)=\left(\alpha_1 t_1+1,\alpha_2 t_1,...,\alpha_d t_1\right)$ and $P(-1)=\left(\beta_1 t_2-1,\beta_2 t_2,...,\beta_d t_2\right)$ such that $\alpha_i,\beta_i\sim \operatorname{N}\left(\mu,\sigma\right)$, the mean distance integral is
$$
S=\int_{-\infty}^\infty dist(P(1),P(-1))\rho\left(\alpha_i,\beta_i\right) \mathrm{d}\alpha_1 \mathrm{d}\alpha_2\cdot\cdot\cdot \mathrm{d}\alpha_d \mathrm{d}\beta_1 \mathrm{d}\beta_2\cdot\cdot\cdot \mathrm{d}\beta_d
$$

Let $S_0$ be the solution of the integral $S$ for two $k$-flats intersecting the unit ball $\mathbb{B}_{c_0}^d$ and $S_1$ be the solution of $S$ for two antipodals $k$-flats tangents the surface of $\mathbb{B}_{c_0}^d$, then by Proposition~\ref{pro:S0toS1} below we get $0<S_0<S_1\leq 2$. Observing that $S_1$ is equals to any antipodal pair of flats that tangents to the surface of $\mathbb{B}_0^d$, w.l.o.g. we use the pair of flats $\left(P(-1),P(1)\right)$. Denote by $S_1$ and $S_\Delta$ the solutions of the integral $S$ for the pairs $(P(-1),P(1))$ and $(P(-\Delta),P(\Delta))$, respectively, Proposition~\ref{pro:S1toSdelta} (below) argues that the density function is invariant while the distance scaling only
in one direction, which implies that a linear change in $\Delta$ cause scaling the mean distance with $\Delta$, which complete the proof.
\qed\end{proof}


\begin{proposition}
\label{pro:prob_drop}
Given two $k-$flats $P(\Delta),P(-\Delta)\in\mathbf{P}$ passing through the points $\left(\Delta,0,...,0\right)$ and $\left(-\Delta,0,...,0\right)$ respectively. Let $X$ denote a random variable of $dist(P(\Delta),P(-\Delta))$. The probability $p$ that $dist(P(\Delta),P(-\Delta))>2$ is strictly greater than zero, i.e., $p=\operatorname{Pr}(X>2)>0$
\end{proposition}

\begin{proof}
From all the non-negative random variables $Y$ that their mean is equal to $S_1\Delta$ and $\operatorname{Pr}(Y\leq 2\Delta)=1$, we would like to find the one that maximizes the probability of $\operatorname{Pr}(Y\leq 2)$, hence, we defined $Y$ to get $2$ if $dist(P_i(\Delta),P_j(-\Delta))\leq 2$ and $2\Delta$ otherwise.  Proposition~\ref{pro:S1toSdelta} (below) implies $S_{\Delta}=S_1\Delta$, using the expectation definition we get $\operatorname{E}(Y)=2q+2\Delta q=S_{\Delta}=S_1\Delta$. Solving the equation and generating a power series expansion for $q$ we got $(1-S_1/2)+(1-S_1/2)\frac{1}{\Delta}+(1-S_1/2)\frac{1}{\Delta^2}+o(\frac{1}{\Delta^3})$. Proposition~\ref{pro:S1isSmall} below implies that $S_1<2$. Substitute this result in the power series expression yields $0<q<1$. Since $q$ is a bound on the probability to accept the flats $P_i(\Delta)$ and $P_j(-\Delta)$, it holds that the probability $p=\operatorname{Pr}(X>2)$ to drop $P_i(\Delta)$ and $P_j(-\Delta)$ is $p\geq 1-q >0$. I.e., we dropped $p$ fraction of the $\binom{n}{2}$ pairs we have got.
\qed\end{proof}

Note that Proposition~\ref{pro:prob_drop} implies that the fraction $\lambda$ of the flats we dropped is at least linear for pair of flats passing through the exact points $\left(\Delta,0,...,0\right)$ and $\left(-\Delta,0,...,0\right)$. The proof also holds for a pair of flats \textit{intersecting the ball} centered at $\left(\Delta,0,...,0\right)$ and $\left(-\Delta,0,...,0\right)$ by adding the ball's radius.

The following propositions were mentioned in the above proofs and appear here to enhance the readability of the text.
\begin{proposition}
\label{pro:2Dintersection}
Let $\ell'$ and $\ell'' $ be two random 2D lines that intersect the unit disk and $p$ be their intersecting point. With probability $1-O(1)$ the distance $r$ from $p$ to the origin is bounded.
\end{proposition}

\begin{proof}
Observing that the maximum distance between the intersecting point $p$ and the origin occur when $\ell'$ and $\ell''$ are tangents to the disk, we will consider only this case. Let $\phi\in\left[0,\pi\right]$ be the intersection angle between the two lines (see Figure~\ref{fig:2d_intersection}). When $\phi = \pi$ the two lines are joined together and $r$ equals $1$ (the disk radius). While reducing $\phi$ toward the zero angle, $r$ is increased toward infinity (i.e., when $\phi\rightarrow 0$ the lines are parallel and $r\rightarrow\infty$). For example, when $\phi = \pi/2$ then $r=\sqrt{2}$, for $\phi=\pi/3$ we get $r=2$ and generally, $r=1/\sin\frac{\phi}{2}$. Since $\phi$ is uniformly distributed over $\left[0,\pi\right]$, we get that with probability $1-\varepsilon$, the distance $r$ is $\leq r_0$, where $\varepsilon=2\arcsin\left(1/r_0\right)$. E.g., with probability $\geq 2/3$ we have $r\leq 2$.
\qed\end{proof}

\begin{figure}[!h]
\centering
\includegraphics[scale=0.3]{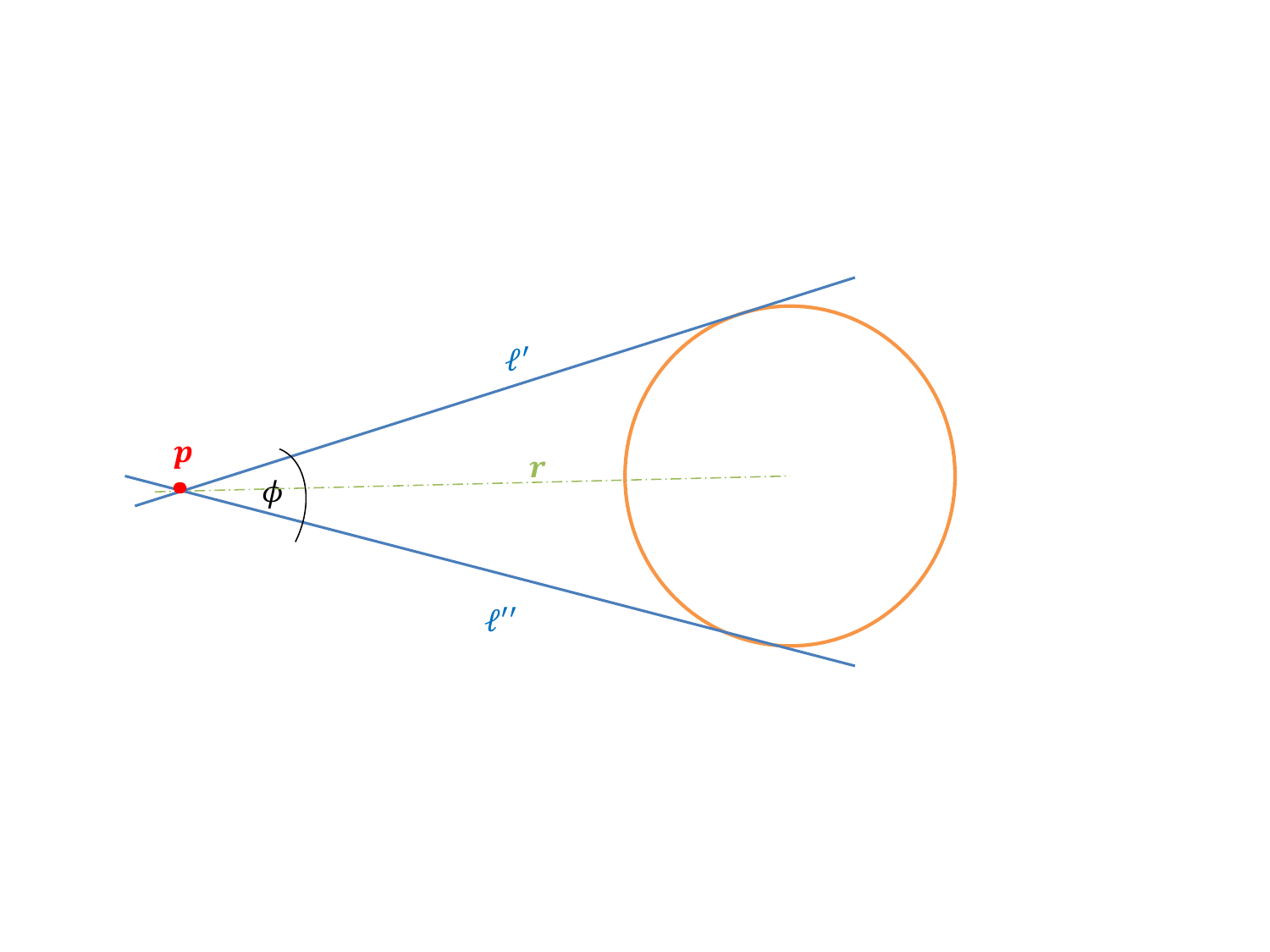}
\caption[Two dimensional pair of flats intersecting a disk]{Two dimensional pair of flats intersecting a disk. While reducing the intersection angle $\phi$ toward the zero angle, the distance $r$ between the intersecting point $p$ and the center of the disk increases towards infinity.}
\label{fig:2d_intersection}
\end{figure}

\begin{proposition}
\label{pro:Dintersection}
Let $P^{c_0}_i$ and $P^{c_0}_j$ be two $k$-flats that intersect the unit ball $\mathbb{B}^d_{c_0}$ and $p$ be their midpoint point. With probability $1-O(1)$ the distance $r$ from $p$ to the origin is bounded.
\end{proposition}

\begin{proof}
Starting with three dimensional space, the relation between the two flats can be expressed using the distance between them, their relative direction (azimuthal angle) and its relative orientation (polar angle). Fixing the orientation, the variation of the direction is described in the 2D case (see proposition~\ref{pro:2Dintersection}). When the two flats' directions cause a small distance between $p$ and the origin, changing the orientation will not increase this distance (but may decrease it). Generally, changing the flat orientation will \textit{increase} the probability that the distance from $p$ to the origin is bound. 

For general $d$, using the same idea, the flats can be represented by a spherical coordinate (i.e., the coordinates consist of a radial coordinate and $d-1$ angular coordinates), implies that the distance between the midpoint and the ball's center is bounded by a probability that increases as $d$ increases, see illustration at Figure~\ref{fig:same_ball}.    
\qed\end{proof}

\begin{figure}[!t]
\centering
\includegraphics[width=2.8in]{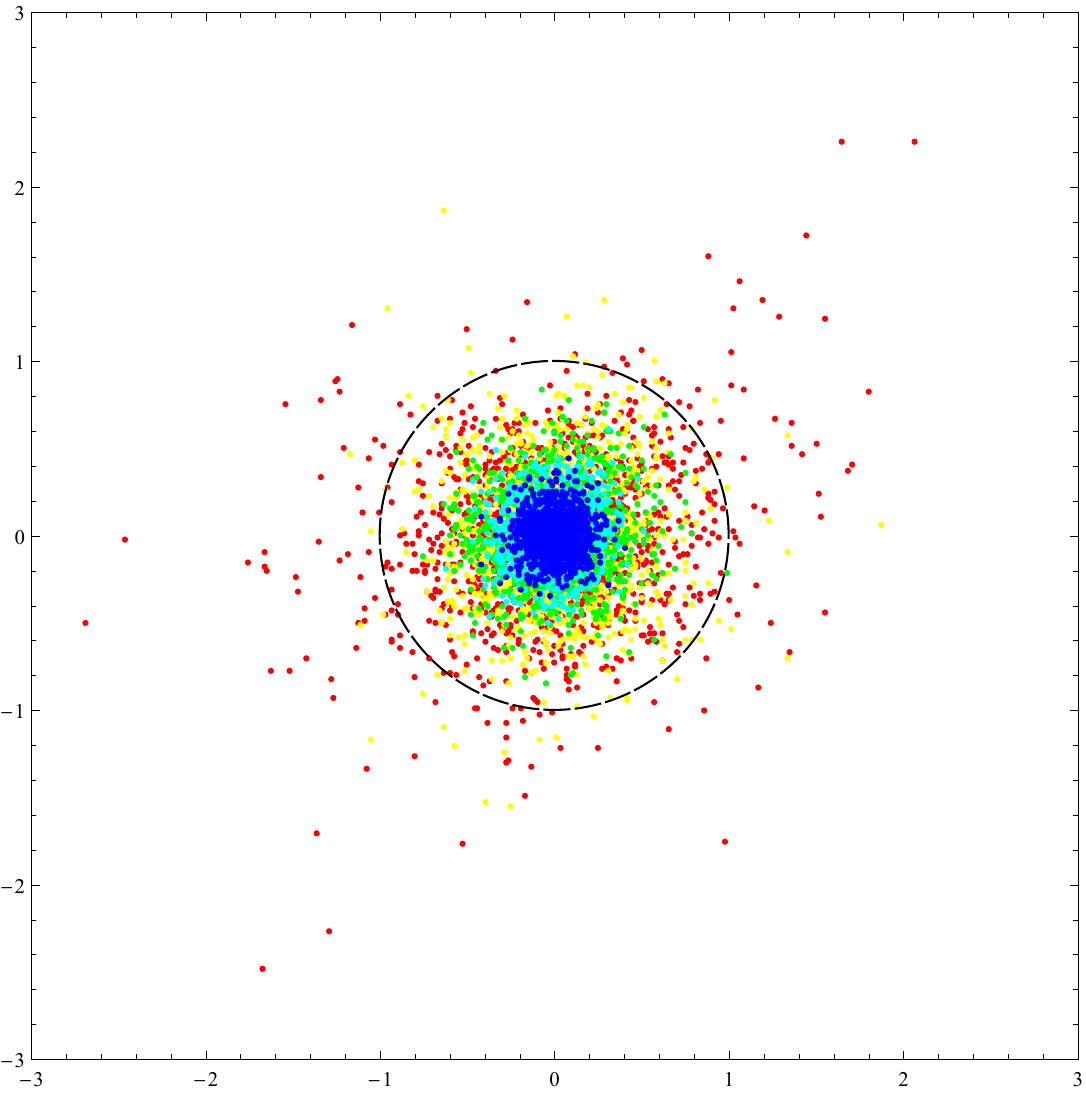}
\caption[The distance between the midpoint and the ball's center]{The distance between the midpoint and the ball's center is decreasing  as $d$ increases. For a unit ball centered at the origin (dashed line), we plot the midpoints (first two coordinates) of a set of $50$ flats with dimension $d=9$ (red dots), $d=18$ (yellow dots), $d=36$ (green dots), $d=72$ (light blue dots) and $d=144$ (blue dots). Note that midpoints from higher dimensions are plotted above those from lower ones. We can observe that most of the midpoints are located inside the unit ball and centered around the origin. Moreover, as the dimension increases, the variance of the location of the midpoints decrease.}
\label{fig:same_ball}
\end{figure}

\begin{proposition}
\label{pro:S0toS1}
Let $S_0$ and $S_1$ be the mean distance integral solutions as defined above, then $0<S_0<S_1\leq 2$.
\end{proposition}

\begin{proof}
Since the degree of the flats is $\leq d/2$ the probability that the flats intersect is $=0$ which implies that $0<S_0$. The mean distance integral $S$ contains a density function $\rho\left(\mu,\sigma\right)$ and a geometric distance $dist\left(\cdot,\cdot\right)$. The density is dependent only on the mean and the variance of the coefficients which are invariant. The distance function get its maximum value for antipodal pair, which implies that $S_0<S_1$. Finally, since the two flats are intersecting the same unit ball, the minimum distance between them is $\leq 2 (\text{ the ball radius})=2$ which implies $S_1\leq 2$ as needed. 
\qed\end{proof}

\begin{proposition}
\label{pro:S1toSdelta}
Let $S_1$ and $S_\Delta$ be the integral solutions as defined above, then $S_\Delta =S_1\Delta$.
\end{proposition}

\begin{proof}
The mean distance of a pair of $k$-flats acts symmetrically on the two pairs $(P(-1),$ $P(1))$ and $(P(-\Delta),P(\Delta))$. Namely, the density function is invariant while the distance scaling only in one direction, which implies a linear change in $\Delta$ in the solution of $S$, i.e., $S_{\Delta}=\Delta S_1$.\qed\end{proof}

\begin{proposition}
\label{pro:S1isSmall}
Let $S_1$ be the integral solutions as defined above, then $S_1 < 2$.
\end{proposition}

\begin{proof}
By its definition, $S_1$ is the mean distance between two flats passing through the points $(-1,0,...,0)$ and $(1,0,...,0)$. Fixing the flat $P(-1)$, we can observe that if $P(1)$ is intersecting the ball $\mathbb{B}_{(-1,0,...,0)}$ then $dist(P(-1),P(1))\leq 1$. Let $\alpha$ denote the probability of this event, i.e. $\alpha=\operatorname{Pr}(P(1)\cap\mathbb{B}_{(-1,0,...,0)}\neq\emptyset)$. To complete the proof, it is enough to prove that $\alpha>0$ (since $S_1
\leq\operatorname{E}(dist(P(-1),P(1)))=1*\alpha+2*(1-\alpha)$). 

Observing that $P(1)\cap\mathbb{B}_{(-1,0,...,0)}$ is a spherical cap with nonzero volume (relatively to the measure of all the flats), one can show that the probability that two random flats passing through $(-1,0,...,0)$ and $(1,0,...,0)$ has distance $\leq 1$ is greater than zero, i.e., $\theta>0$. 
\qed\end{proof}
\section{Algorithm}
\label{sec:algo}
Algorithm~\ref{algo:1} presents the pseudocode for the clustering procedure of a set of $n$ random $k$-flats in $\mathbb{R}^d$. 
\begin{algorithm}
\caption{Data clustering using flats minimum distances}
\label{algo:1}
\begin{algorithmic}[1]
\Require a set $\mathbf{P}$ of $n$ random $k-$flats in $\mathbb{R}^d$, the number of clusters $m$.
\Ensure a set $\mathbf{C}$ of $m$ clusters
\State $\mathbf{p}\gets$ \Call{FindMidpoints}{$\mathbf{P}$}
\State $\mathbf{C}\gets$ \Call{DefineClusters}{$\mathbf{p}$} \Comment{density-based clustering algorithm on the set $\mathbf{p}$, e.g., DBSCAN}
\State$M \gets n/m$ \Comment{threshold for the size of every cluster}
\For{\textbf{each} $c_k\in \mathbf{C}$}
			\If{ $size\left(c_k\right)<M$}
			\State $\mathbf{C}\gets \mathbf{C}\smallsetminus  c_k$
			\EndIf
		\EndFor
\State Return $\mathbf{C}$
\end{algorithmic}
\end{algorithm}

In the first step, we call the procedure \textit{\textsc{FindMidpoints}} to find all the midpoints between all the pairs of flats (using Proposition~\ref{pro:midpoint}) and calculate the distance between every pair (as described in Proposition~\ref{pro:dist}). We save only the midpoints whose corresponding distance is smaller than two. 

\noindent\fbox{%
 \begin{minipage}[b]{\dimexpr\linewidth-2\fboxsep-2\fboxrule\relax}
\begin{algorithmic}[1]
\Procedure{FindMidpoints}{$\mathbf{P}$}
		\State $\mathbf{p}\gets\emptyset$
		\For{\textbf{each} $\left(P_i,P_j\right)\in \mathbf{P}$}
			\State $p_{ij}\gets midpoint\left(P_i,P_j\right)$
			\State $d_{ij}\gets dist\left(P_i,P_j\right)$
			\If{ $d_{ij}\leq 2$}
				\State $\mathbf{p}\gets \mathbf{p}\cup p_{ij}$
			\EndIf
		\EndFor
		\State Return $\mathbf{p}$
	\EndProcedure
	\end{algorithmic}
	\end{minipage}%
}

Using Lemmas~\ref{lem:same} and~\ref{lem:diff} we explore the potential clusters by finding the high density midpoints' locations. We can do this by the use of the classical K-Means-like algorithms. However, since (a small) fraction of the midpoints are `noise', i.e., derived from flats intersecting different balls, we would like to ignore those midpoints. Hence, we recommend using an algorithm that has specialized noise handling such as DBSCAN (Density-Based Spatial Clustering of Applications with Noise) as describe in~\cite{Ester1996}.

Next, we use our assumption (see Section~\ref{sec:pre}) about the equal size of the different clusters and define a threshold $M$ which equals to $n/m$. Now we eliminate all the clusters that their density is low (as defined by the threshold $M$). 

Note that the algorithm outputs a set of clusters $C_{res}=\left\{c_k\right\}$ such that each cluster contains midpoints $c_k=\left\{p_{ij}\right\}$ that indicate that the flats $P_i,P_j$ are in the cluster $c_k$.

\section{Experimental Studies of $k$-Flat Clustering}
\label{sec:exp}
As part of the main theorem proof, Lemma~\ref{lem:diff} tells us what happens when we take the dimensionality to infinity. In practice, it is interesting to know at what dimensionality we anticipate that the flat pairwise projection to midpoints implies good separation to different clusters. In other words, Lemma~\ref{lem:diff} describes some convergence, but does not indicate the convergence rate. We addressed this issue through empirical studies. 

We ran the following experiments using synthetic data set, producing the flats' inputs with normally distributed location and direction, as described in the model assumption. Without loss of generality we choose the balls' center to be $c_1=\left(-100,...,0\right)$ and $c_2=\left(100,0,...,0\right)$ and $k$ (the flats dimension) equals to $d/3$.

Each cluster contain 10 random flats, all together we have 20 random flats. Our algorithm computes the midpoint for all pairs of flats; all together we have 190 center points. See Figure~\ref{fig:midPoints_diff} which shows four different experiments, each done for different dimensions. Those center points are divided into three groups: the first 45 are shown as a red dot close to the center $c_1$. Furthermore, they are close to one another so that the eye cannot distinguish between them. The second group is also comprised of 45 points, shown as a red dot to the right, close to $c_2$. The third group has 100 points, centered around 0 point. Those points are shown in black, with a distance of $>2$. This means that the algorithm rejects all the points in the third group, as was anticipated. The four images illustrate how the variance is decreasing, while increasing the dimension. This illustrates that our algorithm preforms better for higher dimensions. 

\begin{figure*}[!t]
	\centering
	\begin{tabular}{cc}
		\includegraphics[scale=0.5,clip=true]{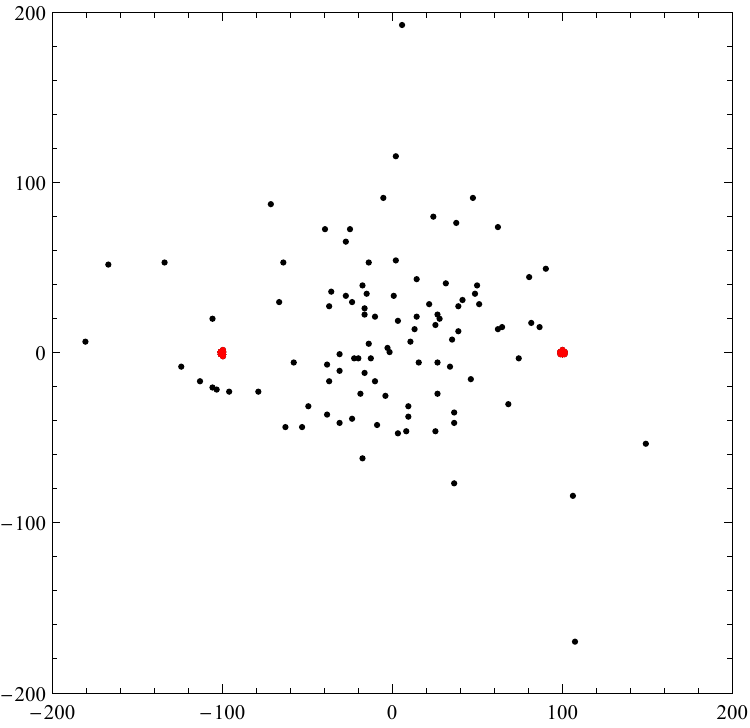} &
		\includegraphics[scale=0.5,clip=true]{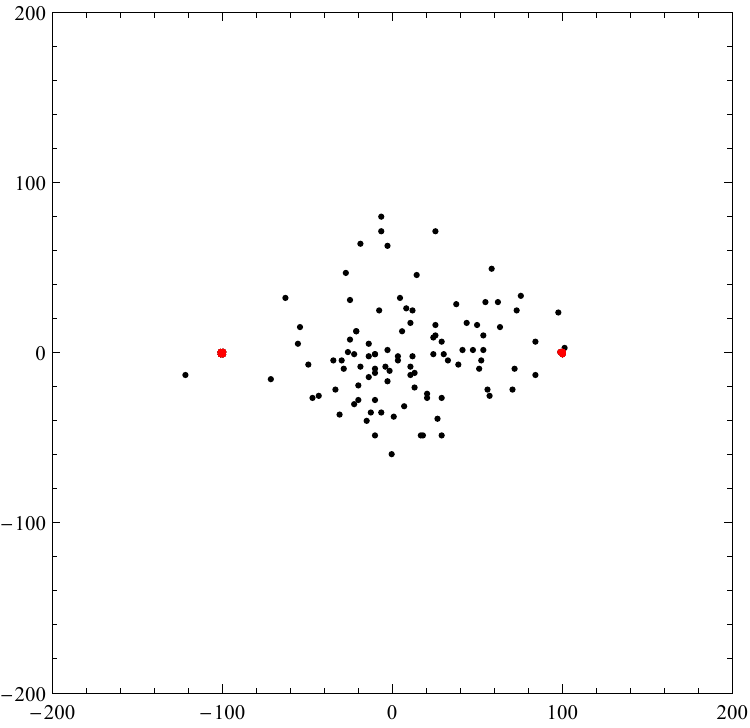} \\
		$d=9,k=3$ & $d=30,k=10$ \\
		\includegraphics[scale=0.5,clip=true]{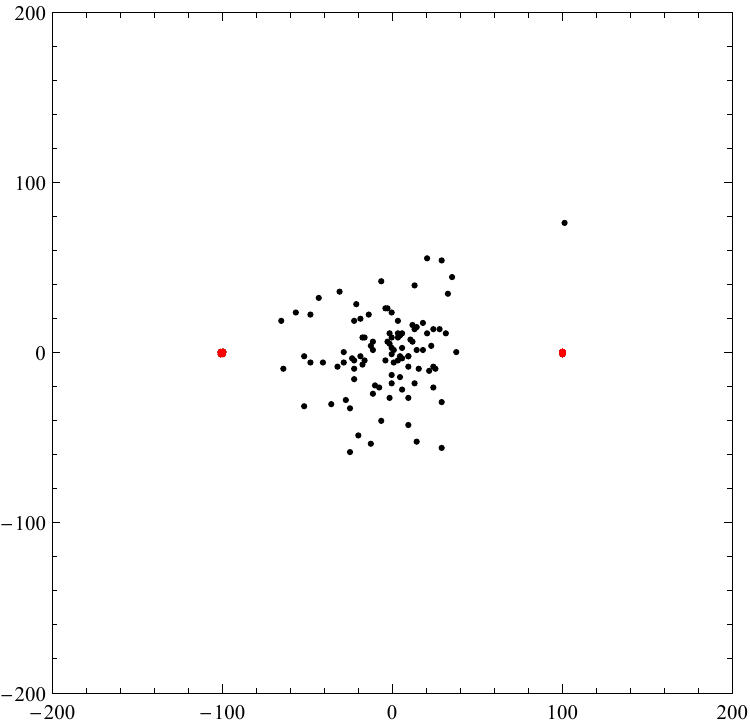} &
		\includegraphics[scale=0.5,clip=true]{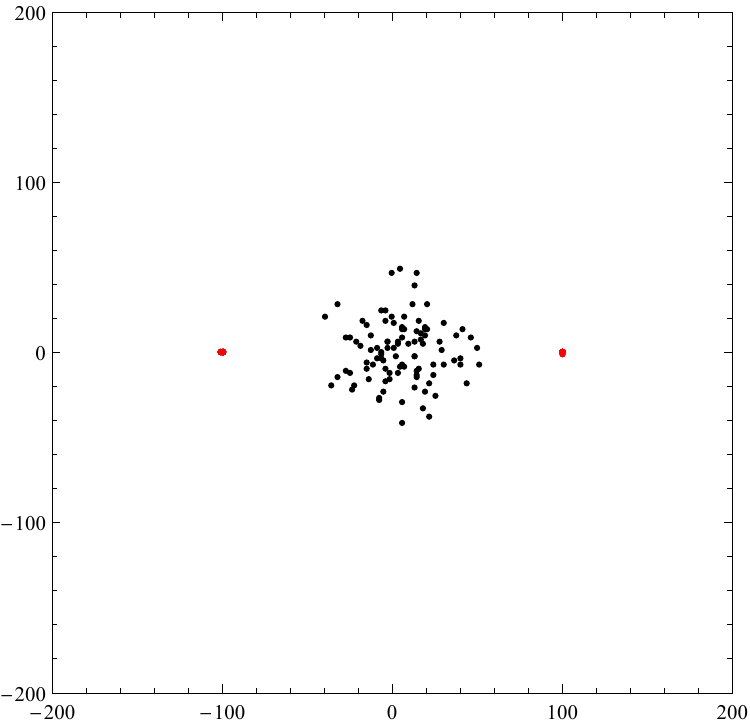} \\
		 $d=60,k=20$ & $d=90,k=30$\\
			\end{tabular}
\caption[Eliminate the irrelevant midpoints]
{\captionsize Given two sets of flats from two clusters located at $\mathbb{B}_{(-100,0,...0)}$ and $\mathbb{B}_{(100,0,...0)}$, the black points are the midpoints of all the pairs and the red points indicate those who left after eliminate flats that their corresponding distance is greater than $2$.
}
	\label{fig:midPoints_diff}
\end{figure*}

\section{Clustering with different group sizes}
\label{sec:diffGroupSize}
Algorithm~\ref{algo:1} we presented above works fine for a set of $m$ clusters, for which each one of them contains the \textit{same} number of flats. We need this assumption to ensure that we will not identify `noisy' midpoints (i.e., midpoints derived from flats intersecting different balls) as a true cluster. In this section we would like to relax the equal size clusters assumption. 

For sufficiently large dimension $d$ we do not need the assumption concerning the equal size of the clusters since we show in Lemma~\ref{lem:diff} that when $d\to\infty$ we will drop all the noisy midpoints (since  w.h.p their distance is larger than two). For a general dimension $d$, we show we drop a fraction $\lambda$ of the noise and argue (as in Proposition~\ref{pro:prob_drop}) that this fraction is at least linear. Hence, instead of assuming equal size clusters we can assume that the difference between the clusters is at most $\lambda$, we will call it \textit{$\lambda$-close size} clusters. Moreover, when the data contains also some very big clusters (i.e., of size greater than the joint number of all the rest) we suggest peeling off these clusters and continuing recursively as described in Algorithm~\ref{algo:2}.

Note that the other model's assumptions hold (see Section~\ref{sec:pre}), i.e., for the given set of $k$-flats $\mathbf{P}$:
\begin{itemize}
\item Two independent random flats are in general position with probability one.
\item $1\leq k\leq \left\lfloor d/2\right\rfloor$.
\item The (unknown) balls $\mathbb{B}^d_{c_1},...,\mathbb{B}^d_{c_m}$ are $\Delta$-\textit{distinct} with probability one.
\item $\mathbf{P}$ is a superset of $m$ groups $\mathbf{P}=\left\{P_1,...,P_m\right\}$, such that every group $P_i\in\mathbf{P}$ contains flats that intersect the ball $\mathbb{B}^d_{c_i}$. Moreover, each flat $P\in P_i$ have \textit{normally distributed} location and direction at the ball $\mathbb{B}^d_{c_i}$.
\end{itemize}

Given a set $\mathbf{P}$ of $n$ flats and the number of clusters $m$, we find the midpoints set using \textsc{FindMidpoints} (similar to Algorithm~\ref{algo:1}). Next, we call the recursive procedure \textsc{RecClustering} that define the potential clusters (\textsc{DefineClusters} as described in Algorithm~\ref{algo:1} above) and check if there exists a big cluster $c_1$ such that its size is greater than $n/2$ (where $n$ denote the number of the remaining flats). If such cluster were explored, we eliminate all the flats belong to it and recursively call the procedure again. Otherwise, we assume that all the clusters are $\lambda$-close-size so the algorithm recognized them in the same way it does in Algorithm~\ref{algo:1}. 

\begin{algorithm}
\caption{Clustering with different sets size}
\label{algo:2}
\begin{algorithmic}[1]
	\Require a set $\mathbf{P}$ of n different $k-$flats in $\mathbb{R}^d$, the number of clusters $m$, the fraction $\lambda$.
	\Ensure a set $\mathbf{C}_{res}$ of clusters
	\State $\mathbf{p}\gets$ \Call{FindMidpoints}{$\mathbf{P}$}
	\State $\mathbf{C}_{res}\gets$ \Call{RecClustering}{$\mathbf{p},Emptyset,m,n$}
	\State Return $\mathbf{C}_{res}$
\end{algorithmic}
\end{algorithm}

\noindent\fbox{%
\begin{minipage}{\dimexpr\linewidth-2\fboxsep-2\fboxrule\relax}
\begin{algorithmic}[1]
	\Procedure{RecClustering}{$\mathbf{p},\mathbf{C}_{res},m,n,\lambda$}
		\State $\mathbf{C}\gets$ \textsc{DefineClusters}$\left(\mathbf{p}\right)$
		\State $c_1\gets \displaystyle\max_{c_k}\left\{\mathbf{C}\right\}$
		\If {$size\left(c_1\right)> n/2$}
				\State $\mathbf{C}_{res}\gets \mathbf{C}_{res}\cup c_1$
				\State $\mathbf{p}\gets \mathbf{p}\smallsetminus \left\{p_{ij}:p_{ik}\in c' \text{ \texttt{or} } p_{kj}\in c' \right\}$
				\State $n\gets n-size\left(c_1\right)$
				\State $\mathbf{C}_{res}\gets$\Call{RecClustering}{$\mathbf{p},\mathbf{C}_{res},m-1,n,\lambda$}
		\Else
			\For{\textbf{each} $c_k\in \mathbf{C}$}
				\If{ $\textsc{size}\left(c_k\right)<\lambda n/m$}
				\State $\mathbf{C}\gets \mathbf{C}\smallsetminus  c_k$
				\EndIf
			\EndFor
		\EndIf
		\State Return $\mathbf{C}_{res}$
	\EndProcedure
\end{algorithmic}
\end{minipage}%
}

Proposition~\ref{pro:ClusteringDiffSize} argue the correctness of the above approach.

\begin{proposition}
\label{pro:ClusteringDiffSize}
Given the set $\mathbf{P}$ of $n$ different $k$-flats in $\mathbb{R}^d$. Let $s_1$ denote the largest set of flats that intersect the same unit ball. When the size of $s_1$ is greater than $n/2$ Algorithm~\ref{algo:2} identify correctly $s_1$'s flats as a true cluster.
\end{proposition}
\begin{proof}
Let $n_1$ be the number of flats in $s_1$. Algorithm~\ref{algo:2} find a cluster $c_1$ contains $\alpha\binom{n_1}{2}$ midpoints that all of them produced by $s_1$ flats' and also a cluster $c_2$ contains $\beta\binom{n_1(n-n_1)}{2}$ midpoints produced by mixture of flats from $s_1$ (i.e., $c_2$ contains the midpoints $\left\{p_{ij}:P_i\in s_1, P_j\notin s_1\right\}$). Lemmas~\ref{lem:same} and~\ref{lem:diff} imply that w.h.p. for a sufficiently large dimension, $\alpha > \beta$. Using the assumption about the size of $s_1$ we get that $\binom{n_1}{2}>\binom{n_1(n-n_1)}{2}$, hence, $c_1$ is the largest size cluster identify correctly (in Line $4$) and the wrong cluster $c_2$ will be dropped (see Line $6$).\qed\end{proof}

\section{Sublinear and distributed algorithms}
\label{sec:subAndDistribute}
Given a set of size $n$ with $k$-flats in $\mathbb{R}^d$, the algorithms we presented above find the distance and the midpoints of every pair in $O\left((kd)^\omega\right)$ time (where $\omega$ is the matrix multiplication complexity), using the least squares method. By doing this for $\binom{n}{2}$ pairs, we presented a $poly(n,k,d)$ running time algorithm. One can achieve polylogarithmic time using sampling. Namely, instead of running the algorithms with the whole $n$ flats set, we apply the algorithms with a sample of $\log n$ flats that were picked uniformly at random. The main reason we can use sampling is due to our assumption about the normally distributed data, namely, that the given set of flats $\mathbf{P}$ is a superset of $m$ groups $\mathbf{P}=\left\{P_1,...,P_m\right\}$, such that every group $P_i\in\mathbf{P}$ contains flats $P\in P_i$ that have \textit{normally distributed} location and direction at the ball $\mathbb{B}^d_{c_i}$. 

Another way to improve efficiency is to execute the algorithm in a distributed fashion. We describe the distributed algorithm in the procedure \textsc{DistributedFindMidpoints} which replace the procedure \textsc{FindMidpoints} of Algorithm~\ref{algo:2}. Given a set of $q$ processors such that each one of them has an access to the whole set of flats $\mathbf{P}$, every processor randomly picks a pair of flats and calculate their midpoints. If the distance between the pair is less than two, the processor saves the midpoint in shared memory (stored in the set $\mathbf{p}$). The processors continue this procedure until enough midpoints were collected as defined by the threshold $\tau$. The clustering process can be done by any of the processors as described in Algorithm~\ref{algo:2}. The correctness of this algorithm follows the birthday paradox that promises that with high probability there will not be an overlap between the processors due to the small fraction of the sampling. 

\noindent\fbox{%
\begin{minipage}{\dimexpr\linewidth-2\fboxsep-2\fboxrule\relax}
\begin{algorithmic}[1]
	\Procedure{DistributedFindMidpoints}{$\mathbf{P},\mathbf{p},\tau$}
			\While{\textsc{Size}$\left(\mathbf{p}\right)< \tau$}
				\State randomly pick a pair of flats $\left(P_i,P_j\right)\in \mathbf{P}$
				\State $p_{ij}\gets midpoint\left(P_i,P_j\right)$
				\State $d_{ij}\gets dist\left(P_i,P_j\right)$
				\If{ $d_{ij}\leq 2$}
					\State $\mathbf{p}\gets \mathbf{p}\cup p_{ij}$
				\EndIf
			\EndWhile
	\EndProcedure
\end{algorithmic}
\end{minipage}%
}   

Note, we can also replace the sequential \textsc{DefineClusters} procedure in Algorithm~\ref{algo:2} with a distributed one, namely, for DBSCAN algorithm there also exists a distributed version (see e.g.,~\cite{Januzaj2004}).
\section{Discussion}
\label{sec:diss}
The probability of flats' intersections appear at different settings in~\cite{Kendall1963} and~\cite{Santalo2004}. Using polar representation~\cite{Kendall1963} measure the probability that $d$ $k$-flats going through a ball, will intersect each other inside the ball. E.g., for $d=2$ and $k=1$, random lines intersecting a disk will intersect each other inside the disk with probability $1/2$ and for $d=3$ and $k=2$, three planes that intersecting a convex region $K$ will have their common point inside $K$ with probability $\pi^2/48$. These results are generalized in~\cite{Santalo2004} for $n$ randomly chosen subspaces $f_{k_i}$ ($i=1,2,...,n$) in $\mathbb{E}^d$, such that $k_1+k_2+...+k_n\geq (n-1)d$, that intersect a convex body $K$. Formalized the probability that $f_{k_1}\cap f_{k_2}\cap ... \cap f_{k_n}\cap K \neq \emptyset$ using the integral:
$
\int\displaylimits_{f_{k_1}\cap f_{k_2}\cap...\cap f_{k_n}\cap K\neq\emptyset} df_{k_1}\wedge f_{k_2}\wedge...\wedge df_{k_n}
$
\cite{Santalo2004} (13.39),(14.2) show that the measure of all $k$-flats $f_k$ that intersect a convex body $K$ in $\mathbb{E}^d$ is $\frac{O_{d-1}\cdot\cdot\cdot O_{d-k-1}}{(d-k)O_{k-1}\cdot\cdot\cdot O_0}$ (where $O_d$ denotes the surface area of the $d$-dimensional unit sphere). Another related result one can extract from~\cite{Santalo2004} work is the probability of a hyperplane $L_{d-1}$ and a line $L_1$ that intersect a ball to have an intersection inside the ball, which equals $1/d$. A detailed description of the above results appear in the following Appendix.

The studies of~\cite{Kendall1963} and~\cite{Santalo2004} consist on polar representation of the data (i.e., the coordinates consisting of a radial coordinate and $d-1$ angular coordinates) which gives high weight to the first coordinate while the weight of the following coordinates decrease (since the coefficients are multiples of sine and cosine). Hence, our assumption on normal distribution over the different coordinates is not fulfilled.

Another direction we examine was to find the unknown balls (that defined the clusters) using the intersection of orthogonal flats. The justification of focus on orthogonal sets comes from the ``curse of dimensionality'' phenomenon~\cite{Bellman1961}, were one manifestation of the ``curse'' is that in high dimensions, almost any two vectors are almost orthogonal~\cite{Rajaraman2011}. Starting with the two dimensional case, we generate a random set of flats intersecting disjoint balls and picked all the almost orthogonal pairs, namely, pairs that their intersecting angle is in $\left[\pi/2 \pm \varepsilon\right]$ (where $\varepsilon$ depend on $n$- the number of flats). Interestingly, as presenting in Figure~\ref{fig:orthogonal}, we can describe the pairwise intersection by distinguish two sets: those passing through the same ball, and those arise from different balls. The first set concentrated around the original balls center, while the second set create a structural figure, corresponding to the relative geometrical positioning of the original balls. This geometric structure might help in defining the origin unit balls, but the exact definition of $\varepsilon$ and the generalization to higher dimensions should be examined in further research. 
\begin{figure}[!t]
\centering
\includegraphics[scale=0.5]{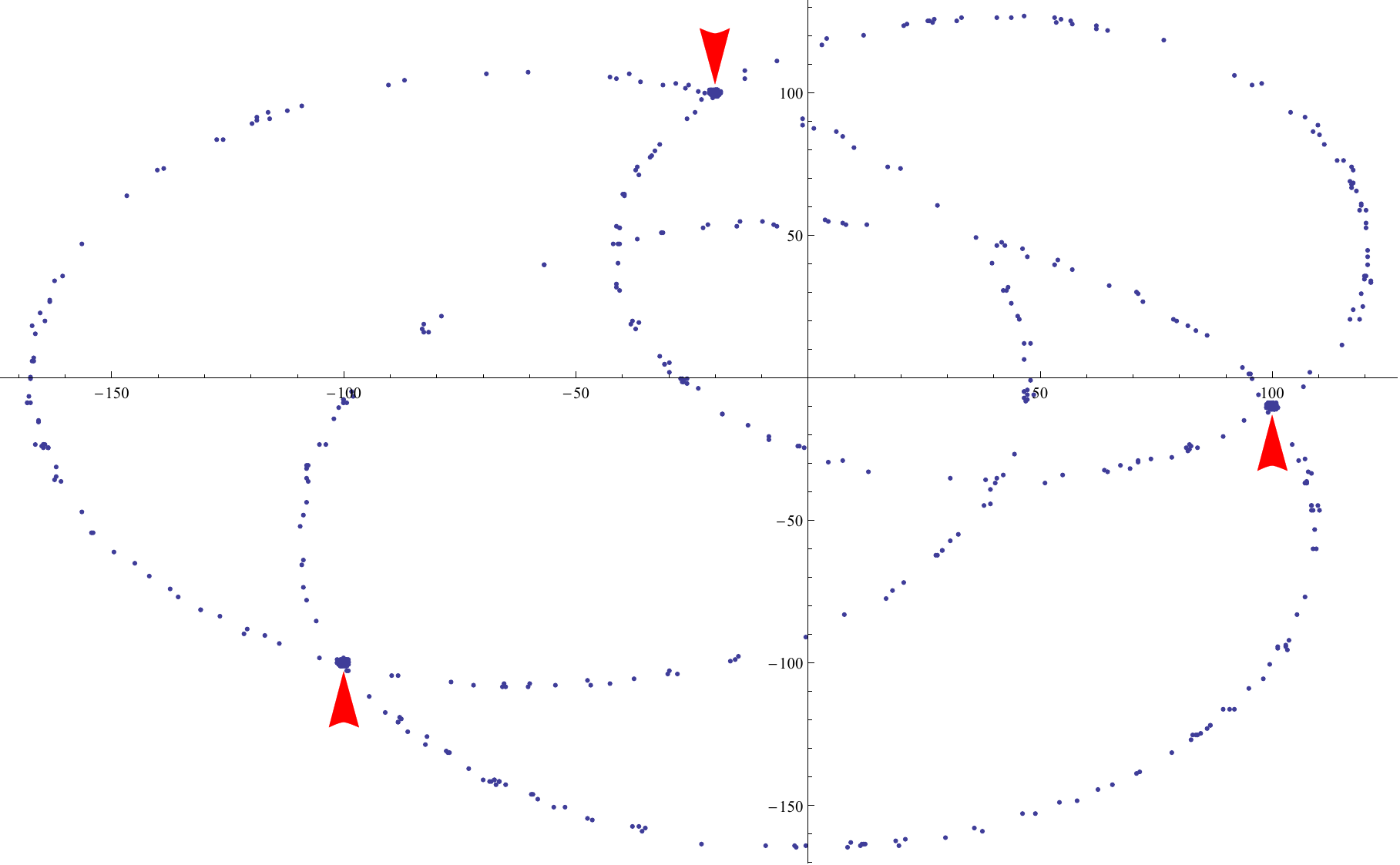}
\caption[Almost orthogonal flats pairwise intersection]{Pairwise intersection of two dimensional almost orthogonal flats from three disjoint balls. Given a set of $1000$ random lines passing through the three unit balls $\mathbb{B}^2_{\left(-100,-100\right)},\mathbb{B}^2_{\left(100,-10\right)},\mathbb{B}^2_{\left(-20,100\right)}$, we plot the intersecting point of every almost orthogonal pair, i.e., pair of lines that their intersecting angle is in $\left[\pi/2\pm \varepsilon\right]$, where $\varepsilon=0.001$. One can identify the original three centers (see red arrows) by the points concentrating around their regions. In addition, the rest of the points located on the boundary of the circles connecting the centers.}
\label{fig:orthogonal}
\end{figure}

Canas et al.~\cite{Canas2012} study the problem of estimating a manifold from random $k$-flats. Given collections of $k$-flats in $\mathbb{R}^d$
their (Lloyd-type) algorithm, analogously to $k$-means, aims at finding the set of $k$-flats that minimizes an empirical reconstruction over the whole collection. Although they also deal with the input of $k$-flat, their framework and goals are different from ours, and specifically, impractical for the clustering task.

The distance between pairs of $k$-flats as well as measuring the geometry of the midpoints was studied in~\cite{Schulte2014} and generalized at~\cite{Hug2015}. Although these papers consider the probabilistic aspects of the flats intersections, as we do, they focus only on stationary processes (such as Poisson processes) that do not satisfy the uniform and Gaussian distributions that we assume here.

As mentioned above, Lee $\&$ Schulman~\cite{Lee2013} presented algorithms and hardness results for clustering general $k$-flats in $\mathbb{R}^d$. After proving that the exponential dependence on $k$ (the internal dimension of the flat) and $m$ (the number of clusters) is inevitable they suggest an algorithm which runs in time exponential in $k$ and $m$ but is linear in $n$ and $d$. Their theoretical results are based on the assumption that the flats are axis-parallel. Our model overcomes their exponential bounds due to the randomized assumption.

\section{Conclusion}
\label{sec:con}
The analysis of incomplete data is one of the major challenges in the scope of big data. Typically, data objects are represented by points in $\mathbb{R}^d$, we suggest that the incomplete data is corresponding to affine subspaces. With this motivation we study the problem of clustering $k$-flats, where two objects are similar when the Euclidean distance between them is small. The study presented a simple clustering algorithm for $k$-flats in $\mathbb{R}^d$, as well as studied the probability of pair-wise intersection of these objects. 

The key idea of our algorithm is to formulate the pairs of flats as midpoints, which preserves distance features. This way, the geometric location of midpoints that arise from the same cluster, identify the center of the cluster with high probability (as shown in Lemma~\ref{lem:same}). Moreover, we also show (Lemma~\ref{lem:diff}) that when the dimension $d$ is big enough, the corresponding distance of flats that arise from different clusters approach the mean distance of the cluster's center. Using this, we can eliminate the irrelevant midpoints with high probability. 

For low dimensions, we did not identify the exact probability that we dropped all the irrelevant flats (i.e., those that arise from different clusters), however, we do show that we eliminate a linear fraction $\lambda$ of those irrelevant flats. In addition, using experimental results, we support our claim that the algorithm works well in low dimensions as well.

Finally, we show we can achieve a polylogarithmic running time using sampling; we also illustrate a distributed version of the algorithm.

Future work includes proving that $\lambda\to 1$ for a general dimension $d$ (we show this only for $d\to\infty$). Obtaining this result will make our algorithm practical to any mixture size of clusters.


\bibliographystyle{abbrv}
\bibliography{references}
\section{Appendix: The probability of flats intersection}
\label{apdx:chap4}

The probability of flats intersection appear at different settings in~\cite{Kendall1963} and~\cite{Santalo2004}. Due to ``Bertrand Paradox" (see explanation at~\cite{Kendall1963} Introduction) the most natural coordinates to use for the description of flats in the $d$ Euclidean space,$\mathbb{E}^d$, are the polar coordinates. Starting with the two dimensional space, a line on the plane is determined by its distance $p$ from the origin and the angle $\theta$ of the normal with the $x$ axis. The equation of the line is 
$$
x\cos \theta + y\sin \theta - p =0
$$
The measure of the set of all lines $L_1$ intersecting a bounded convex set $K$ is~\cite{Santalo2004}(3.12)
\begin{eqnarray}
m(L_1,L_1\cap K \neq \emptyset)= \int\displaylimits_{L_1\cap K\neq \emptyset}p d\theta = L=2\pi
\end{eqnarray}
where $L$ is the length of $\partial K$ (perimeter of $K$, for the disk its equals to $2\pi$).

The measure for two random chords of $K$ to intersect inside $K$ is~\cite{Kendall1963}(3.9)
\begin{eqnarray}
\int x dp d\theta=2\pi A
\end{eqnarray}
where $A$ is the area of $K$. 

Since the measure of each line that intersecting $K$ is $L$ and they are taken as independent, the appropriate measure for pair of lines is $L^2$. This implies that the probability for random lines intersecting a disk to intersect each other inside the disk is 
\begin{eqnarray}
p=\frac{2\pi A}{m(L_1,L_1\cap K \neq \emptyset)^2}=\frac{2\pi \pi}{(2\pi)^2}=\frac{1}{2}.
\end{eqnarray}
This result is fixed while changing the radius of the disk. (Note: the probability that \textit{all} the intersection points lie inside $K$ is $< \frac{n!}{\left(2n\right)!}\left(\frac{bL}{2}\right)^n$ (where $b$ is the maximal value of the curvature of $\partial K$), see~\cite{Sulanke1965}).

\subsection*{Random planes in the 3D space.}
At the three dimensional space, the original set of flats might appear as lines or planes. We continue here by assuming a set of planes (the probability of lines and mixture of lines and planes appear as part of the general case). The appropriate definition for planes is given by the polar equation~\cite{Kendall1963}(4.1):
\begin{eqnarray}
x\sin \theta \cos \phi + y \sin \theta \sin \phi + z \cos \theta = p
\end{eqnarray}
and the element of measure is
\begin{eqnarray}
\sin \theta d\theta d\phi dp
\end{eqnarray}
where $0\leq \theta\leq \pi$ and $0\leq\phi\leq 2\pi$. 

For calculating the probability that three planes intersecting $K$ have their common point inside $K$, we need the value of the integral
$$
m(L_{2_i},L_{2_j},L_{2_\ell};L_{2_i}\cap L_{2_j}\cap L_{2_\ell}\cap K\neq\emptyset)=\int\displaylimits_{L_{2_i}\cap L_{2_j}\cap L_{2_\ell}\cap K\neq\emptyset} dL_{2_i}\wedge L_{2_j}\wedge dL_{2_\ell}
$$
Like in the planar case, first we extract the measure $M$ of all planes $L_2$ that meeting a convex region $K$, which is
\begin{eqnarray}
m(L_2;L_2\cap K \neq \emptyset) = \int\displaylimits_{L_2\cap K\neq\emptyset}dL_2 = 4\pi
\end{eqnarray}
The proof of this is given by Minkowski~\cite{Kendall1963} (see Section 4.7). 

Now we calculate the measure that three planes $L_{2_i},L_{2_j},L_{2_\ell}$ that meet $K$, intersect each other inside $K$. Suppose two of the planes intersecting inside $K$, denote the intersection length by $L$, i.e., $L_{2_i}\cap L_{2_j}\cap K=L$. The measure of all planes which intersect $L$ is $\pi L$~\cite{Kendall1963}(4.3).
The integral of $L$ over all positions of one of these planes is $\frac{1}{2}\pi^2 A$~\cite{Kendall1963}(4.7),where $A$ is the area of intersection of the other plane. In turn, the integral of $A$, the area of intersection over all intersecting planes is $2\pi V$. Hence, the measure of all such triples is $\pi^4V$, and the required probability is
\begin{eqnarray}
p=\frac{\pi^4 V}{m(L_2;L_2\cap K \neq \emptyset)^3}=\frac{\pi^4 \frac{4}{3} \pi}{(4\pi)^3}=\pi^2/48
\end{eqnarray}
Changing the radius of the sphere inversely proportional to the probability $p$, since the radius $R$ at the numerator is power of $3$ (part of the volume $V=4/3\pi R^3$), but at the denominator $R$ have power of $6$.

\subsection*{Random $r$-planes in $\mathbb{E}^d$.}
Given $n$ randomly chosen subspaces $L_{r_i}$ ($i=1,2,...,n$), such that $r_1+r_2+...+r_n\geq (n-1)d$, that intersect d-dimensional ball $\mathbb{B}^d$. We would like to find the probability that $L_{r_1}\cap L_{r_2}\cap ... \cap L_{r_n}\cap \mathbb{B}^d \neq \emptyset$. Namely, to solve the integral:
$$
m(L_{r_1},L_{r_2},...,L_{r_n};L_{r_1}\cap L_{r_2}\cap...\cap L_{r_n}\cap K\neq\emptyset)=\int\displaylimits_{L_{r_1}\cap...\cap L_{r_n}\cap K\neq\emptyset} dL_{r_1}\wedge L_{r_2}\wedge...\wedge dL_{r_n}
$$
Mimic the way we use in the low dimensions, we have to calculate the measure of all $r$-planes $L_r$ that intersect $\mathbb{B}^d$, and also to find out the measure that all the intersecting of set is interior to $\mathbb{B}^d$.  
Let $O_d$ denote the surface area of the d-dimensional unit sphere and $\kappa_d$ denote the volume of the n-dimensional unit ball. Their values are: 
\begin{eqnarray}
O_d=\frac{2\pi^{(d+1)/2}}{\Gamma((d+1)/2)} & ; & \kappa_d=\frac{O_{d-1}}{d}=\frac{2\pi^{d/2}}{d\Gamma(d/2)}
\end{eqnarray}
where $\Gamma$ is the Gamma function. For instance, $O_0=2,O_1=2\pi,O_2=4\pi,O_3=2\pi^2$.

The measure of \textit{all} $r$-planes $L_r$ that intersect $\mathbb{B}^d$ appear at~\cite{Santalo2004} (13.39),(14.2):
\begin{eqnarray}
m(L_r,L_r\cap \mathbb{B}^d \neq \emptyset)= \frac{O_{d-1}\cdot\cdot\cdot O_{d-r-1}}{(d-r)O_{r-1}\cdot\cdot\cdot O_0}
\end{eqnarray}
Santalo~\cite{Santalo2004} also show that 
\begin{eqnarray}
p(L_p\cap L_q\cap \mathbb{B}^d\neq\emptyset ; p + q = d) = \frac{p!q!O_{d-1}\kappa_d}{\left(d-1\right)!O_{p-1}O_{q-1}}\nonumber\\
p(L_p\cap L_q\cap \mathbb{B}^d\neq\emptyset ; p + q > d) = \frac{2\left(p-1\right)!\left(q-1\right)!O_{2d-p-q+1}}{\left(p+q-d-1\right)!\left(d-1\right)!O_{d-p+1}O_{d-q+1}}
\end{eqnarray}
This result will help us while we use the intersection of pairs of flats to locate the Ball's center.

Another result we can extract from~\cite{Santalo2004} work is the probability of a hyperplane $L_{n-1}$ and a line $L_1$ that intersect a ball having an intersection inside the ball:
\begin{eqnarray}
p(L_{1},L_{n-1};L_{1}\cap L_{n-1}\cap \mathbb{B}^d\neq\emptyset) = 1/n
\end{eqnarray}
This result can be useful for records such that at one of the record all exclude one of the coordinate is missing ($L_{n-1}$), and for the second record, only one coordinate is missing ($L_1$).


\end{document}